\documentclass[envcountsame]{llncs}

\usepackage{mymacros}
\usepackage{url} 
\usepackage[ruled,lined]{algorithm2e}
\usepackage{graphicx}
\usepackage{amsmath}

\usepackage{verbatim}
\usepackage{color}
\usepackage{pdfsync} 
\usepackage{amsfonts}
\usepackage[english]{babel} 

\begin{document}


\title{ Analysing Cycloids using Linear Algebra}
\author{R\"udiger Valk}
\institute{University of Hamburg, Department of Informatics\\
Hamburg, Germany \\
  \email{ruediger.valk@uni-hamburg.de}}
\maketitle

\begin{abstract}
Cycloids are particular Petri nets for modelling processes of actions or events. They belong to the fundaments of Petri's general systems theory and have very different interpretations, ranging from Einstein's relativity theory and  elementary information processing gates to the modelling of interacting sequential processes. This article contains previously unpublished proofs of cycloid properties using linear algebra. 
\end{abstract}



\section{Introduction}\label{sec-intro}
Cycloids have been introduced  by C.A. Petri in \cite{Petri-NTS} in the section on physical spaces, using as examples firemen carrying the buckets with water to extinguish a fire, the shift from Galilei to Lorentz transformation and the representation of elementary logical gates like Quine-transfers.
Based on  formal descriptions of cycloids in \cite{Kummer-Stehr-1997} and \cite{fenske-da} a more elaborate formalization is given in \cite{Valk-2019}.
\begin{figure}[htbp]
 \begin{center}
        \includegraphics [scale = 0.31]{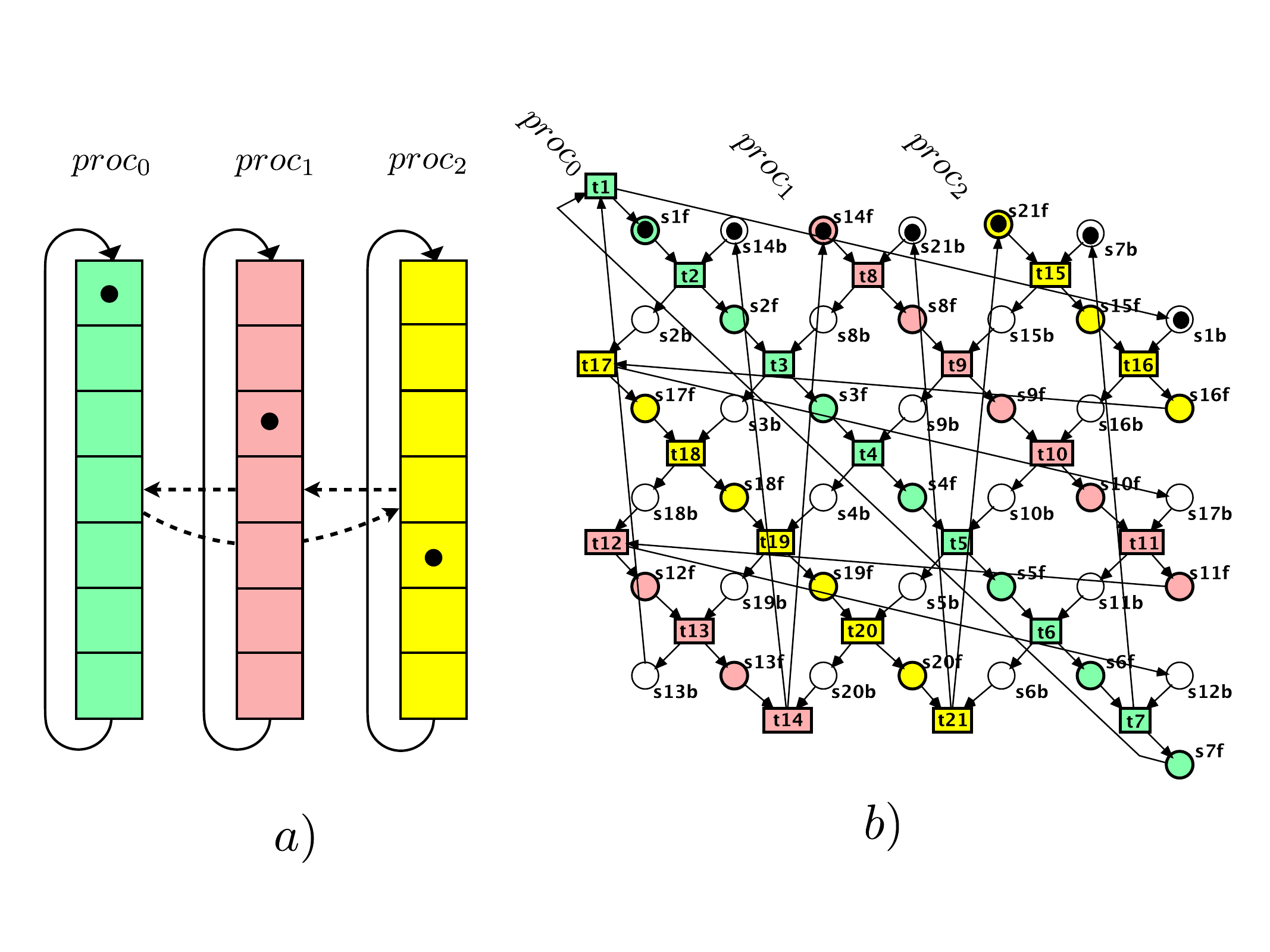}
        \caption{Three sequential processes synchronized by single-bit channels, }
        \label{c-4-3-3-3-b}
      \end{center}
\end{figure}

Cycloids are structures that are defined with methods of discrete mathematics, which makes proofs sometimes not very descriptive. It was therefore a great step forward that a method was introduced in \cite{Valk-2020} that allows proofs to be carried out with the help of linear algebra. This method is called \emph{Cycloid Algebra}. 
Three theorems are proved in this article using Cycloid Algebra, namely
a) on the equivalence of transitions with respect to the cycloid folding,
b)  on isomorphisms of cycloids and 
c) on the minimal length of cycles with respect to the grafic structure of a cyloid.


To give an application for the theory, as presented in this article,  consider a distributed system of a finite number of circular and sequential processes. The processes are synchronized by uni-directional one-bit channels in such a way that they behave like a circular traffic queue when folded together. To give an example, Figure \ref{c-4-3-3-3-b}a) shows three such sequential circular processes, each of length $7$. In the initial state the control is in position $1$, $3$ and $5$, respectively. The synchronization, realized by the connecting channels, should be such as the three processes would be folded together. This means, that the controls of $proc_0$ and $proc_1$ can make only one  step until the next process $proc_2$ makes a step itself, while
the control of $proc_2$ can make  two steps until $proc_0$ makes a step. Following  \cite{Valk-2020} this behaviour is realized by the cycloid of 
Figure \ref{c-4-3-3-3-b}b) modelling the three processes 
by the transition sequences $proc_0 =$ [\textbf{t1 t2 $\cdots$ t7}], as well as $proc_1 =$ [\textbf{t8 t9 $\cdots$ t14}] and $proc_2 =$ [ \textbf{t15 t16 $\cdots$ t21}]. The channels are represented by the safe places connecting these processes. By this example the power of the presented theory is shown, since the rather complex net is  unambiguously determined by the parameters 
 $\mathcal{C} ( \alpha, \beta, \gamma, \delta ) = \mathcal{C}(4,3,3,3)$. A next question could be, how to change the cycloid when the parameters of $\beta = 3$ processes of process length  $p  = 7$ should be changed to a different value, say the double $p  = 14$. As will be explained in a forthcoming article, the theory returns even three cycloids, namely  $ \mathcal{C}_1( 4,3,10,3 ) $, $ \mathcal{C}_2( 4,3,6,6 ) $ and $ \mathcal{C}_3( 4,3,2,9) $. However,  as follows from Theorem \ref{symmetry} these three solutions are isomorphic. The flexibilty of the model is also shown by the following additional example. By doubling in  $ \mathcal{C}( 4,3,3,3 ) $ the value of $\beta$ we obtain the cycloid $ \mathcal{C}( 4,6,3,3 ) $, which models a distributed system of three circular sequential processes, each of  length $p=10$. However, different to the examples above, each process contains 
 \textbf{two }  control tokens. 
 Translated to the distributed model, in the initial state each of the three sequential processes contains two items, particularly $proc_0$ in positions $0$ and $5$ in the circular queue of length $10$, 
  $proc_1$ in positions $1$ and $6$ and $proc_2$ in positions $3$ and $8$.
 The present article is part of a general project to investigate all such features of cycloids to make them available for  Software Engineering.


We recall some standard notations for set theoretical relations.
If $R\subseteq A \!\times\! B$ is a relation and $U \subseteq A$ then 
$R[U]:= \{b\,|\,\exists u \in U: (u,b)\in R\}$ is the \emph{image} of $U$ and $R[a]$ stands for $R[\{a\}]$. 
$R^{-1}$ is the \emph{inverse relation} and $R^+$ is the \emph{transitive closure} of $R$ if $A=B$.
Also, if $R\subseteq A \!\times\! A$ is an equivalence relation then $\eqcl[R]{a}$
 is the \emph{equivalence class} of the quotient $A/R$ containing $a$.
Furthermore  $\Nat$, $\Natp$, $\Int$ and $\Real$ denote the sets of integers,  positive integer, integer and real numbers, respectively.
For integers: $a|b$ if $a$ is a factor of $b$.
The $modulo$-function is used in the form 
$a \,mod\,b = a - b \cdot \lfloor \frac{a}{b} \rfloor$, which also holds for negative integers $a \in \Int$.
In particular, $-a\,mod\,b = b-a$ for $0<a\leq b$.


\section{Petri Space and Cycloids}  \label{sec-cycloids}

We define (Petri) nets as they will be used in this article. 

 \begin{definition}[\cite{Valk-2019}] \label{def-net} 
 As usual, a net $ \N{} = (S, T, F)$ is defined by non-empty, disjoint sets 
 $S$ of places and $T$ of transitions, connected by a flow relation 
 $F \subseteq (S \cp T) \cup (T \cp S)$ and $X := S \cup T$.
A transition $t \in T$ is \emph{active} or \emph{enabled} in a marking $M \subseteq S$ if $\; ^{\ndot} t \subseteq M \, \land \, t^{\ndot} \cap M = 
\emptyset$\footnote{With the condition $t^{\ndot} \cap M = \emptyset$ we follow Petri's definition, but with no impacts in this article.}.
In this case we obtain $M \stackrel{t}{\rightarrow}M'$ if $M' = M \backslash^{\ndot} t \cup t^{\ndot}$, where 
$^{\ndot} x := F^{-1}[x], \; x^{\ndot} := F[x] $
denotes the input and output elements of an element $x \in X$, respectively. 
 $\stackrel{*}{\rightarrow}$ is the reflexive and transitive closure of $\rightarrow$. 
\end{definition}

\begin{figure}
	\begin{center}
	\includegraphics [scale = 0.22]{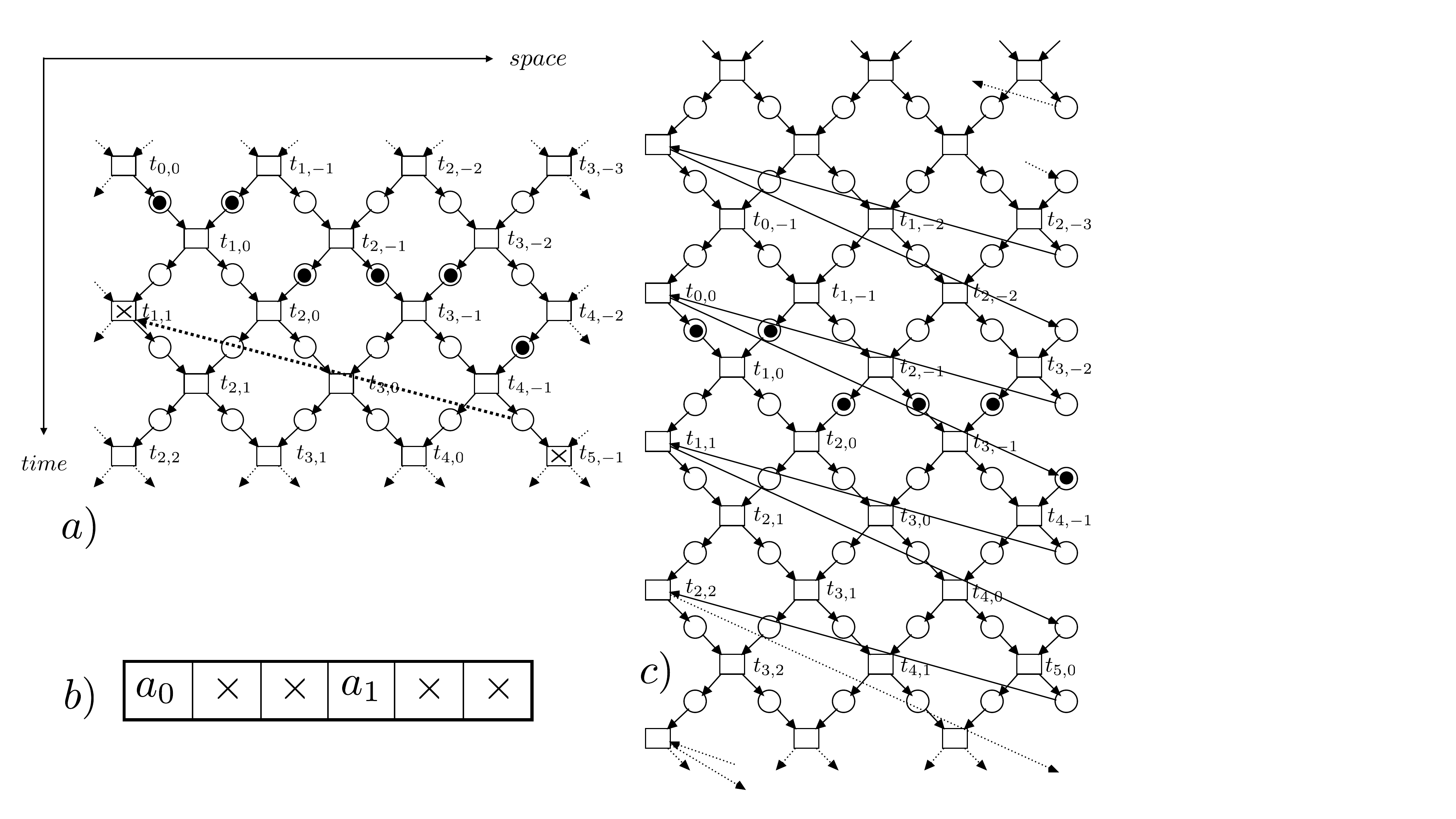}
		\caption{ a) Petri space, b) circular traffic queue and c) time orthoid.}
	\label{petrispace-key}
	 \end{center}
\end{figure}

Petri started with an event-oriented version of the Minkowski space which is called Petri space now. 
Contrary to the Minkowski space, the Petri space is independent of an embedding into $\Int \times \Int$.
It is therefore suitable for the modelling in transformed coordinates as in non-Euclidian space models.
However, the reader will wonder that we will apply linear algebra, for instance using equations of lines.
This is done only to determine the relative position of points.
It can be understood by first topologically transforming and embedding the space into $\Real \times \Real$, calculating the position and then transforming back into the Petri space.
Distances, however, are \underline{not} computed with respect to the Euclidean metric, but by counting steps in the grid of the Petri space, like Manhattan distance or taxicab geometry.

For instance, the transitions of the Petri space might model the moving of items in time and space in an unlimited way.
To be concrete, a coordination system is introduced with arbitrary origin (see Figure~\ref{petrispace-key} a).
The occurrence of transition $t_{1,0}$ in this figure, for instance, can be interpreted as a step of a traffic item (the token in the left input-place) in both space and time direction.
It is enabled by a gap or co-item (the token in the right input-place).
 Afterwads the traffic item can make a new step by the occurrence of transition $t_{2,0}$.
By the following definition the places obtain their names by their input transitions
(see Figure~\ref{P-space+FD} b).

\begin{figure}
	\begin{center}
	\includegraphics[width=0.8\textwidth]{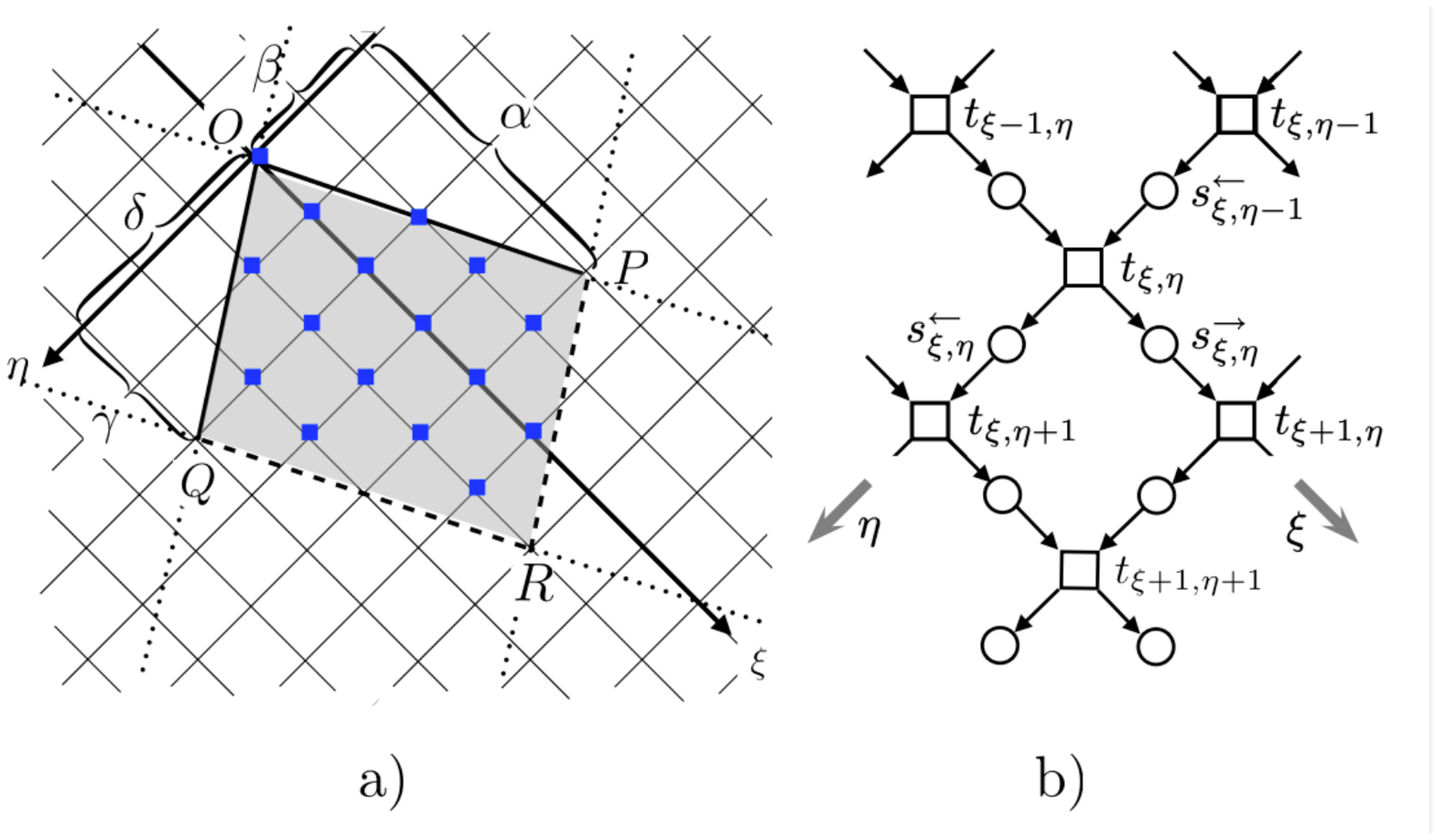}
		\caption{a)  Fundamental parallelogram of $\mathcal{C}(4,2,2,3)$ and  b) Petri space.}
		\label{P-space+FD}
	\end{center}
\end{figure}

\begin{definition} [\cite{Valk-2019}] \label{petrispace}
A $Petri \; space$ 
is defined by the net 
$\PS{1} := (S_1, T_1, F_1)$ \ where
$S_1 = \GSvw[1]\cup \GSrw[1], \;\GSvw[1]  =$ 
   $\col{\gsvw{\xi,\eta}}{\xi,\eta \in \Int},$ 
$\;\GSrw[1] = \col{\gsrw{\xi,\eta}}{\xi,\eta \in \Int }, \GSvw[1] \cap \GSrw[1] = \emptyset $, 
 $ T_1 =$ $ \col{t_{\xi,\eta}}{\xi,\eta \in \Int }, F_1 =$ $ \col{(t_{\xi,\eta},\gsvw{\xi,\eta})}{\xi,\eta \in \Int } \cup \col{(\gsvw{\xi,\eta},t_{\xi+1,\eta})}{\xi,\eta \in \Int } \cup $ 
 $\col{(t_{\xi,\eta},\gsrw{\xi,\eta})}{\xi,\eta \in \Int } \cup \col{(\gsrw{\xi,\eta},t_{\xi,\eta+1})}{\xi,\eta \in \Int }$ (cutout in Figure~\ref{P-space+FD} b).
 $\GSvw[1]$ is the set of \emph{forward places} and $\GSrw[1]$ the set of \emph{backward places}.
 $ \prenbfw{t_{\xi,\eta}}:= \gsvw{\xi-1,\eta}$ is the forward input place of $t_{\xi,\eta}$ and in the same way 
 $ \prenbbw{t_{\xi,\eta}}:= \gsrw{\xi,\eta-1}$,
 $ \postnbfw{t_{\xi,\eta}} := \gsvw{\xi,\eta}$ and
 $\postnbbw{t_{\xi,\eta}}:= \gsrw{\xi,\eta}$ (Figure~\ref{P-space+FD} b).

\end{definition}

In two steps, by a twofold folding with respect to time and space, Petri defined the cyclic structure of a cycloid.
 One of these steps is  
 a folding $f$ with respect to space with $f(i,k)=f(i+\alpha,k-\beta)$, fusing all points $(i,k)$ of the Petri space with $(i+\alpha,k-\beta)$ where $i,k \in \Int, \alpha,\beta \in \Natp$ (\cite{Petri-NTS}, page 37).
 While Petri gave a general motivation, oriented in physical spaces, we interpret the choice of $\alpha$ and $\beta$ by our model of traffic queues.

We assume that our model of a circular traffic queues has six slots containing two items $a_0$ and $a_1$ as shown in Figure~\ref{petrispace-key} b).
These are modelled in Figure~\ref{petrispace-key} a) by the tokens in the forward input places of $t_{1,0}$ and $t_{3,-1}$.
The four co-items (the empty slots in Figure \ref{petrispace-key} b) ) are represented by the tokens in the backward input places of $t_{1,0}, t_{2,0}$ and $t_{3,-1}, t_{4,-1}$.
By the occurrence of $t_{1,0}$ and $t_{2,0}$ the first item can make two steps, as well as the second item by the transitions $t_{3,-1}$ and $t_{4,-1}$, respectively.
Then $a_1$ has reached the end of the queue and has to wait until the first item is leaving its position.
Hence, we have to introduce a precedence restriction between the transitions $t_{1,0}$ and $t_{5,-1}$.
This is done by fusing the transitions 
 $t_{5,-1}$ and the left-hand follower $t_{1,1}$ of $t_{1,0}$ ,  which are marked by a cross in Figure~\ref{petrispace-key} a).  This is implemented by the dotted arc in the same figure.
To determinate $\alpha$ and $\beta$ we set 
 $(5,-1) =(1+\alpha,1-\beta)$ which gives 
 $5 = 1+ \alpha$ or $\alpha = 4$ and 
 $-1 = 1- \beta$ or $\beta = 2$.
By the equivalence relation
 $t_{\xi,\eta} \equiv t_{\xi + 4,\eta -2} $ we obtain the structure in Figure~\ref{petrispace-key} c). 
The resulting still infinite net is called a \emph{time orthoid} (\cite{Petri-NTS}, page 37), as it extends infinitely in temporal future and past.
The second step is a folding with $f(i,k)=f(i+\gamma,k+\delta)$ with $\gamma,\delta \in \Natp$ reducing the system to a cyclic structure also in time direction. 
As shown in \cite{Valk-2020} an equivalent cycloid for the traffic queue of Figure~\ref{petrispace-key} b) has the parameters $ ( \alpha, \beta, \gamma, \delta ) = (4,2,2,2)$.
To keep the example more general, 
in Figure~\ref{P-space+FD} a) the values 
 $ ( \alpha, \beta, \gamma, \delta ) = (4,2,2,3)$ are chosen.
 In this representation of a cycloid, called \emph{fundamental parallelogram}, the squares of the transitions as well as the circles of the places are omitted.
 All transitions with coordinates within the parallelogram belong to the cycloid including those on the lines between $O,Q$ and $O, P$, but excluding those of the points $Q,R,P$ and those on the dotted edges between them.
 All parallelograms of the same shape, as indicated by dotted lines outside the fundamental parallelogram are fused with it.

\begin{definition} [\cite{Valk-2019}] 
\label{cycloid}
A \emph{cycloid} is a net \ $ \zyk( \alpha, \beta, \gamma, \delta ) = (S, T, F)$, defined by parameters 
 \ $ \alpha, \beta, \gamma, \delta \in \Natp$, by a quotient \cite{SR:SemNN:87} of the Petri space \ $\PS{1} := (S_1, T_1, F_1)$ \ 
with respect to the equivalence relation
$\mo\zykaeq 
\subseteq X_1 \cp X_1 $
with $X_1 = S_1 \cup T_1$,
 $ \mo\zykaeq[\GSvw[1]] \subseteq \GSvw[1], 
 \mo\zykaeq[\GSrw[1]] \subseteq \GSrw[1],
 \mo\zykaeq[T_1] \subseteq T_1,$ 
 $ x_{\xi,\eta} \zykaeq x_{\xi+m\alpha+n\gamma,\,\eta-m\beta+n\delta} $
for all $ \xi, \eta, m, n \in \Int $\,, $ X = X_1/_\zykaeq $,
 $ \eqcl[\zykaeq]{x} \mb{F} \eqcl[\zykaeq]{y} \: \Leftrightarrow
\exists\,{x'\in\eqcl[\zykaeq]{x}}\,\exists\,y' \in \eqcl[\zykaeq]{y}: \,x' F_1 y' $
 \ for all $x, y \in X_1 $. 
The matrix $\mathbf{A} = \begin{pmatrix} \alpha & \gamma \\ -\beta & \delta \end{pmatrix} $ is called the matrix of the cycloid.
Petri denoted the number $|T|$ of transitions as the area $A$ of the cycloid and proved in \cite{Petri-NTS} its value to $|T| =A =\alpha\delta+\beta\gamma$ which equals the determinant $A = det(\mathbf{A})$.
The embedding of a cycloid in the Petri space is called \emph{fundamental parallelogram} 
(see Figure~\ref{P-space+FD} a).
\end{definition}


\section{Equivalence and Isomorphisms}  \label{sec-equi}

For proving the equivalence of two points in the Petri space the following procedure\footnote{The algorithm is implemented under \url{http://cycloids.de}.} is useful.

\begin{theorem}[\cite{Valk-2020}] \label{parameter}
Two points $\vec{x}_1, \vec{x}_2\in X_1$ are equivalent $\vec{x}_1 \equiv \vec{x}_2$ if and only if
 for the difference $\vec{v} := \vec{x_2}-\vec{x_1}$
 the parameter vector 
 $\pi(\vec{v}) = \frac{1}{A} \cdot\mathbf{B} \cdot \vec{v}$ has integer values, where $A$ is the area and 
 $\mathbf{B} = \begin{pmatrix} \delta & -\gamma \\ \beta & \alpha \end{pmatrix}$. \\
 In analogy to Definition \ref{cycloid} we obtain 
 $\vec{x}_1 \equiv \vec{x}_2 \Leftrightarrow$ 
 $\exists \;m, n \in \Int: \vec{x_2}-\vec{x_1} =\mathbf{A}\begin{pmatrix} m \\ n \end{pmatrix}$.
\end{theorem}

\begin{proof}
For  $\vec{x}_1 := (\xi_1,\eta_1),
 \vec{x}_2 := (\xi_2,\eta_2), 
\vec{v} := \vec{x}_2 - \vec{x}_1$
from Definition \ref{cycloid} we obtain \ in vector form: 
$\vec{x}_1 \equiv \vec{x}_2 \Leftrightarrow 
\exists \, m,n \in \Int: \begin{pmatrix} \xi_2 \\ \eta_2  \end{pmatrix} =  
\begin{pmatrix} \xi_1 +m\alpha +n\gamma \\ \eta_1-m\beta + n\delta  \end{pmatrix}$
$\Leftrightarrow$ \\
$\exists \, m,n \in \Int:$
$\vec{v}=\begin{pmatrix} \xi_2 -\xi_1\\ \eta_2 - \eta_1 \end{pmatrix} = \begin{pmatrix}  m\alpha+n\gamma\\ -m\beta+n\delta  \end{pmatrix} = \begin{pmatrix} \alpha & \gamma \\ -\beta & \delta \end{pmatrix} \begin{pmatrix} m \\ n  \end{pmatrix} = \mathbf{A}\begin{pmatrix} m \\ n  \end{pmatrix}$ 
$\Leftrightarrow$
$\begin{pmatrix} m \\ n  \end{pmatrix} = \mathbf{A}^{-1}\vec{v}$ 
$\in \Int \times \Int$ . 
It is well-known that $\mathbf{A}^{-1} = \frac{1}{det(\mathbf{A})}\mathbf{B}$ 
if $det(\mathbf{A}) > 0$ (see any book on linear algebra). 
The condition $det(\mathbf{A}) = A = \alpha\delta + \beta\gamma >0$ is satisfied by the definition of a cycloid.
\qed \end{proof}


Since constructions of cycloids may result in different but isomorphic forms the following theorem is important.
A method using linear algebra together with the matrices $\mathbf{A}$ in Definition \ref{cycloid} or $\mathbf{B}$ in Definition \ref{parameter} is called a \emph{Cycloid Algebra} method.
We give here a proof using this approach, which was not yet known when the article \cite{Valk-2019} had been published.

\begin{theorem}[\cite{Valk-2019}]\label{symmetry}
The following cycloids are net isomorphic (Definition \ref{def-net}) to $\mathcal{C}(\alpha,\beta,\gamma,\delta) $:\\
     $\;\;\;$ a) $\mathcal{C}(\alpha,\beta,\gamma -  \alpha,\delta+\beta) $ if  $\gamma >  \alpha$, \\
      $\;\;\;$ b) $\mathcal{C}(\alpha,\beta,\gamma +  \alpha,\delta- \beta) $ if $\delta >  \beta$.\\
       $\;\;\;$ c) $\mathcal{C}(\beta,\alpha,\delta, \gamma)$. (The $symmetric \; cycloid$ of $\mathcal{C}(\alpha,\beta,\gamma,\delta) $.)
\end{theorem}

\begin{proof} 
In all the three cases we give a bijection on the Petri space, which is a congruence with respect to equivalence.
Let be $ \mathcal{C} = \mathcal{C}( \alpha, \beta, \gamma, \delta ) $ with matrix
$\mathbf{A}$ (Definition \ref{cycloid}) and the vector $\overrightarrow{mn} := (m,n) \in \Int^2$.\\
a) and b): The bijection is the identity map and we prove that the equivalence relation of 
$ \mathcal{C}_1 = \mathcal{C}_1( \alpha, \beta, \gamma \pm \alpha, \delta \mp \beta) $ 
with 
  matrix 
 $\mathbf{A}_1 = \begin{pmatrix} \alpha &\; \gamma \pm \alpha \\ -\beta & \;\delta \mp \beta \end{pmatrix}$  
remains unchanged:
 $\mathbf{A}_1 \cdot \overrightarrow{mn} = 
 \mathbf{A} \cdot \overrightarrow{mn} + \begin{pmatrix} 0 & \;\; \pm \alpha \\ 0 & \;\; \mp \beta \end{pmatrix}\cdot \overrightarrow{mn} =
 \mathbf{A} \cdot \overrightarrow{mn} + \begin{pmatrix} \pm n \cdot \alpha \\ \mp n\cdot\beta \end{pmatrix} =
 \mathbf{A} \cdot \overrightarrow{mn} + \mathbf{A} \cdot \begin{pmatrix} \pm n \\ 0 \end{pmatrix} =$ 
$ \mathbf{A} \cdot \begin{pmatrix} m \pm n \\ n \end{pmatrix} $.
 Hence, the by Theorem \ref{parameter} b) the equivalence relations of $ \mathcal{C}$ and $ \mathcal{C}_1 $ are the same, since $m$ and $n$ are integers iff $m \pm n$ and $n$ are integers.\\
 c): We denote $ \mathcal{C}_2 = \mathcal{C}_2( \beta,\alpha,\delta,\gamma ) $ with matrix
  $\mathbf{A}_2 = \begin{pmatrix} \beta &\; \delta\\ -\alpha & \gamma \end{pmatrix}$.
  Using the sets $X$ and $X_2$ of $ \mathcal{C}$ and $ \mathcal{C}_2$, respectively (Definition \ref{def-net}),
  the isomorphism is defined by $\varphi(x_{\xi,\eta}) := x_{\eta+\beta,\xi-\alpha}$. 
  Obviously, $\varphi$ is injective and surjective. \\
  In the following we use the indices as coordinates of thge points in the Petri space and write \\
  $\varphi \begin{pmatrix} \xi\\ \eta \end{pmatrix} =  \begin{pmatrix} \eta+\beta\\ \xi-\alpha \end{pmatrix} $. It remains to prove that $\varphi$ is a congruence, i.e. 
  $$\begin{pmatrix} \xi\\ \eta \end{pmatrix} \equiv \begin{pmatrix} \xi_1\\ \eta_1 \end{pmatrix} \; \Rightarrow \; \varphi\begin{pmatrix} \xi\\ \eta \end{pmatrix} \equiv \varphi\begin{pmatrix} \xi_1\\ \eta_1 \end{pmatrix} $$
 For the precondition of this implication we have  by Theorem \ref{parameter} b)
 $\begin{pmatrix} \xi\\ \eta \end{pmatrix} \equiv \begin{pmatrix} \xi_1\\ \eta_1 \end{pmatrix} \; \Leftrightarrow \;  \begin{pmatrix} \xi-\xi_1 \\ \eta - \eta_1 \end{pmatrix} = 
 \mathbf{A}\begin{pmatrix} m\\ n  \end{pmatrix} = \begin{pmatrix} \alpha \cdot m + \gamma \cdot n \\ -\beta \cdot m + \delta \cdot n  \end{pmatrix} $  for some $m,n \in \Int$.
 We use this term to prove the conclusio:
 $\varphi\begin{pmatrix} \xi \\ \eta  \end{pmatrix} \equiv \varphi\begin{pmatrix} \xi_1\\ \eta_1  \end{pmatrix}  
 \; \Leftrightarrow \; 
 \begin{pmatrix} \eta+\beta \\ \xi-\alpha  \end{pmatrix}  - \begin{pmatrix} \eta_1+\beta \\ \xi_1-\alpha 
 \end{pmatrix} =
 \begin{pmatrix} \eta - \eta_1 \\ \xi - \xi_1  \end{pmatrix} =
 $
 $\mathbf{A}_2 \begin{pmatrix} m' \\ n' \end{pmatrix}  = 
 \begin{pmatrix} \beta \cdot m' + \delta \cdot n'\\ -\alpha \cdot m' + \gamma \cdot n'  \end{pmatrix} $ 
 for some $m',n' \in \Int$. Using the precondition the conclusio holds by setting 
 $m' := -m$ and $n' := n$.
\end{proof}        
 
In plane geometry, a shear mapping is a linear map that displaces each point in a fixed direction, by an amount proportional to its signed distance from the line that is parallel to that direction and goes through the origin\footnote{\url{https://en.wikipedia.org/wiki/Shear_mapping}}. 
For a cycloid $ \mathcal{C}( \alpha, \beta, \gamma, \delta ) $ the corners of its fundamental parallelogram have the coordinates 
$O = \begin{pmatrix} 0 \\ 0 \end{pmatrix}, 
P = \begin{pmatrix} \alpha \\ -\beta \end{pmatrix},
R = \begin{pmatrix} \alpha+\gamma \\ \delta-\beta \end{pmatrix}$ and $
Q = \begin{pmatrix} \gamma \\ \delta \end{pmatrix} $.
 Comparing them with the corners $O',P',R',Q'$ of the transformed cycloid
$\zyk(\alpha,\beta,\gamma + \alpha,\delta- \beta) $ of Theorem \ref{symmetry} b) we observe $O' = O, P' = P, Q' = \begin{pmatrix} \gamma + \alpha\\ \delta -\beta\end{pmatrix}=R$ and the lines $\overline{Q,R}$ and 
$\overline{Q',R'}$ are the same. Therefore the second is a shearing of the first one.
This is shown in Figure\footnote{The figure has been designed using the tool \url{http://cycloids.adventas.de}.}~\ref{2328-2362} for the cycloids $ \mathcal{C}( 2,3,2,8), \mathcal{C}( 2,3,4,5) $ and $ \mathcal{C}( 2,3,6,2) $. 
\begin{figure}[htbp]
	\begin{center}
		\includegraphics [scale = 0.31]{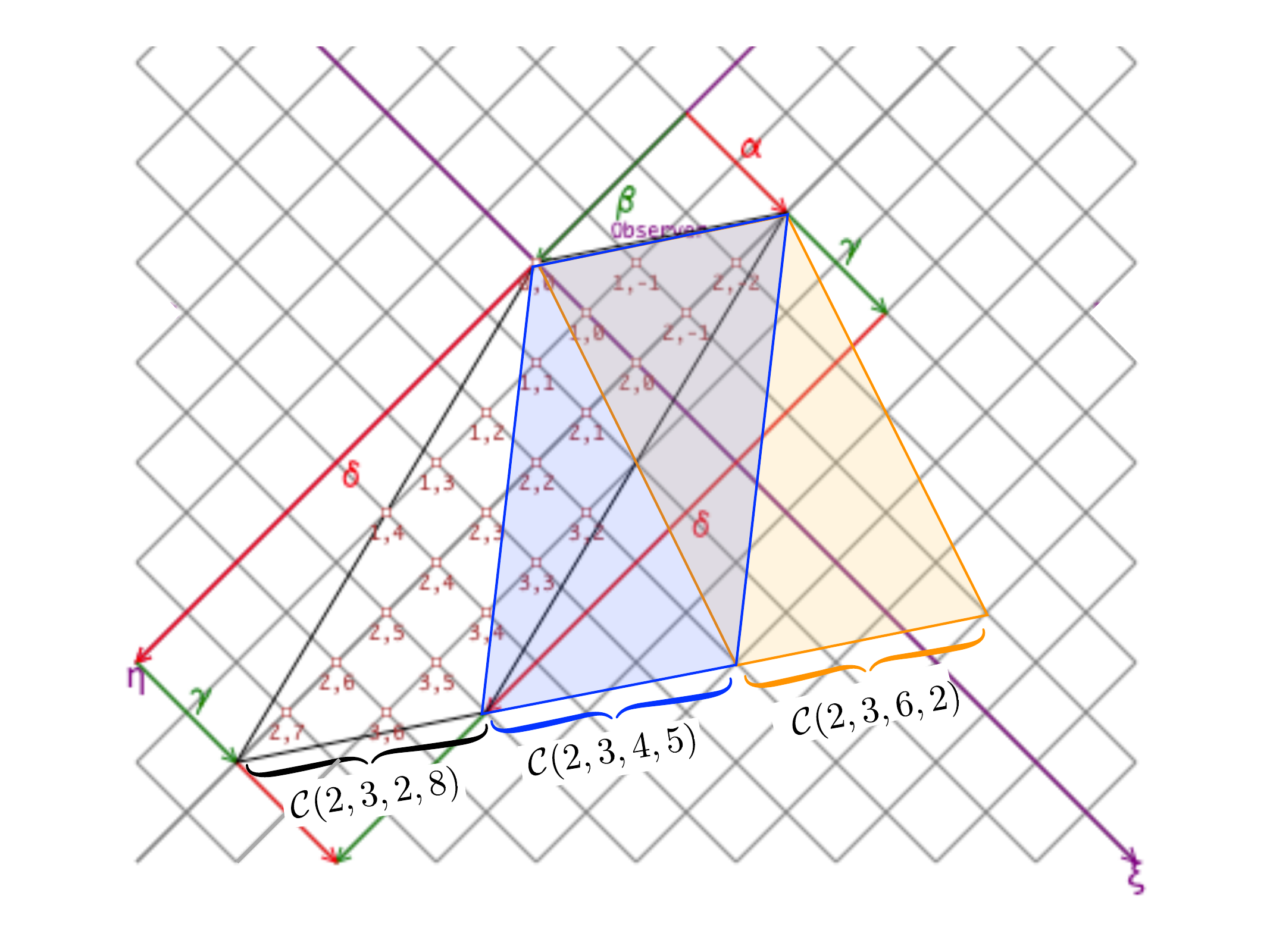}
		\caption{A shearing from  $ \mathcal{C}( 2,3,2,8)   $ to $ \mathcal{C}( 2,3,6,2) $.}
		\label{2328-2362}
	\end{center}
\end{figure}
When applying the equivalences of Theorem \ref{symmetry} the parameters $\gamma$ and $\delta$ 
are changed which leads to the following definition of $\gamma\delta$-reduction equivalence.

\begin{lemma} [\cite{Valk-2019}] \label{normal form}
For any cycloid $\zyk(\alpha,\beta,\gamma,\delta) $ there is a minimal cycle containing the origin $O$ in its fundamental parallelogram representation.
\end{lemma}


\section{The Minimal Length of a Cycle}\label{sec-mincycel}
For the next Theorem from \cite{Valk-2019}, we give a proof which follows the same concept, but is general and more formal. The version from  \cite{Valk-2019} does not cover all cases, but the special case for $\gamma = \delta$ is still valid (see case c) of the following theorem). This subcase was important for the applications to regular cycloids in \cite{Valk-2020}. Part a) of the theorem applies the Cycloid Algebra (Theorem \ref{parameter}). The minimization over two parameters $i$ and $j$ is reduced to one parameter by showing a dependance of to $i$ from $j$ in part b). By restricting to particular cases in the remaining cases no minimum operator is needed.


\begin{theorem}\label{minimal cycles}
The minimal length of a  cycle of a cycloid $\zyk = \zyk(\alpha,\beta,\gamma,\delta)$ is 
$cyc(\alpha,\beta,\gamma,\delta) =  cyc$, where
\begin{itemize}
       \item [a)] 
                      $ cyc = min \{u+v \;|\begin{pmatrix} u \\ v  \end{pmatrix} = \mathbf{A} \cdot 
                      \begin{pmatrix} i \\ j  \end{pmatrix}, \; i \in \Int, \;j  \in \Nat , \; u  \geq 0, \;	v  \geq 0\} $
       \item [b)] $ cyc = min \{j\cdot (\gamma+ \delta)+i\cdot(\alpha-\beta) \;|\;j  \in \Nat, \: i = \left\{
	                     \begin{array}{lll}
		                                    \lfloor\frac{j \cdot \delta}{\beta}\rfloor  &  \textrm{if} &\alpha \leq \beta  \\
		                                   -\lfloor\frac{j \cdot\gamma}{\alpha}\rfloor  & \textrm{if} & \alpha > \beta   
	                        \end{array}
                     \right\} 
                     \} $ \\
                     The value of $j$ is bounded: $j  \leq \frac{A}{\gamma} $ if $\alpha  \leq \beta$ and $j  \leq \frac{A}{\delta} $ otherwise.
       \item [c)]  
                     $cyc = \gamma + \delta + 
                     \; \left\{
	                     \begin{array}{lll}
		            \lfloor\frac{\delta}{\beta}\rfloor (\alpha - \beta) &  \textrm{if} &\alpha \leq \beta  \;\;\textrm{and}  \;\; \gamma  \geq \delta\\
		            -\lfloor\frac{\gamma}{\alpha}\rfloor (\alpha - \beta) & \textrm{if} & \alpha > \beta   \;\;\textrm{and} \;\; \gamma   \leq \delta
	                \end{array}
                     \right\} $ 
         \item [d)] $cyc = \gamma + \frac{\delta}{\beta}\cdot \alpha = \frac{A}{\beta} \;\;\textrm{if}  \;\;  \alpha  \leq \beta 
         \;\;\textrm{and} \;\; \zyk \;\; \textrm{is regular} \;\; (\textrm{i.e.}\;\; \beta|\delta )$
         \item [e)]   $cyc = \delta + \frac{\gamma}{\alpha}\cdot \beta = \frac{A}{\alpha} \;\;\textrm{if}  \;\;  \alpha  > \beta 
         \;\;\textrm{and} \;\; \zyk \;\; \textrm{is co-regular} \;\; (\textrm{i.e.}\;\; \alpha|\gamma )$          
         \end{itemize}           
\end{theorem}

\begin{proof}
a) 
With respect to paths and cycles in the fundamental parallelogram and by Lemma \ref{normal form} it is sufficient to consider paths  starting in the origin $O$.
Such a cycle of the cycloid corresponds to a path with positive length from $O$ to an equivalent point $\vec{x}$ in the Petri space.
From Theorem \ref{parameter}  we obtain with $\vec{x_2} = \vec{x}$ and $\vec{x_1} = (0,0) $  the necessary and suffient condition 
$\exists \;i,j \in \Int : \vec{x} =  \mathbf{A} \cdot \begin{pmatrix} i \\ j  \end{pmatrix}$ with $\lnot (i = 0 \land j=0)$. If $\vec{x}= \begin{pmatrix} u \\ v  \end{pmatrix} $  
then $u+v > 0$ is the length of the path from the origin $O$ to the endpoint of $\vec{x}$. $u+v$ should be a minimum to obtain $cyc$. However, some choices of $\vec{x}$ can be excluded. There is no path from $O$ to 
$(u,v)$ if $u < 0$ or $v < 0$. Therefore $j \leq 0$ can be excluded in $\begin{pmatrix} u \\ v  \end{pmatrix} = \mathbf{A} \cdot 
                      \begin{pmatrix} i \\ j  \end{pmatrix} = \begin{pmatrix} \alpha & \gamma \\ -\beta & \delta \end{pmatrix}\cdot\begin{pmatrix} i \\ j \end{pmatrix} =
                      \begin{pmatrix} i\cdot \alpha + j \cdot \gamma \\ -i \cdot \beta+ j \cdot \delta  \end{pmatrix} $. 
This is true by the following proof by contradiction:  assume $j   \leq 0$. \\
Case 1: If $i \geq 0$ then $ v = -i\cdot\beta+j\cdot\delta < 0$ in contradiction to the condition $v \geq 0$.\\
Case 2: If $i < 0$ then $ u = i\cdot\alpha+j\cdot\gamma < 0$ in contradiction to the condition $u \geq 0$.\\

b)
We first consider the case $\alpha \leq \beta$ and prove $ \lfloor\frac{j \cdot \delta}{\beta}\rfloor $ if $\gamma  \geq \delta$. 
Denote the cutpoint of the line $\overline{QR}$ with the $\xi$-axis by $A_1$ (the line cannot be in parallel to the $\xi$-axis). Next, in a similar way, for $j  \geq 1 $ the endpoint of the vector $j\cdot (\gamma,\delta)$ is denoted by $Q_j$, including $Q_1 = Q$. (See Figure \ref{cyc-theorem} for the cases $j \in \{ 1,2,3 \}$.)
Furthermore we name the cutpoint of the line through $Q_j$ and  the endpoint of  $j\cdot (\gamma,\delta)+ (\alpha, -\beta)$ with the $\xi$-axis by $A_j$.
On this line the points $\vec{x} = j \cdot \begin{pmatrix} \gamma \\\delta \end{pmatrix} + i \cdot \begin{pmatrix} \alpha \\ - \beta \end{pmatrix} = \begin{pmatrix} u \\ v \end{pmatrix} $  are situated which define the value of $cyc = u+v$. By the condition  $v  \geq 0$ we obtain $j \cdot \delta - i \cdot \beta  \geq 0$ which is
\begin{equation}\label{GL1}
i  \leq \frac{j \cdot \delta}{\beta} 
 \end{equation}
Next we derive an expression for $i$ in dependance of $j$ by proving
that increasing the value of $i$ does not increase the distance to the origin (while the condition
$\eta \geq 0$ is not violated when going $\beta$ steps in direction $-\eta$).
\begin{figure}[htbp]
	\begin{center}
		\includegraphics [scale = 0.31]{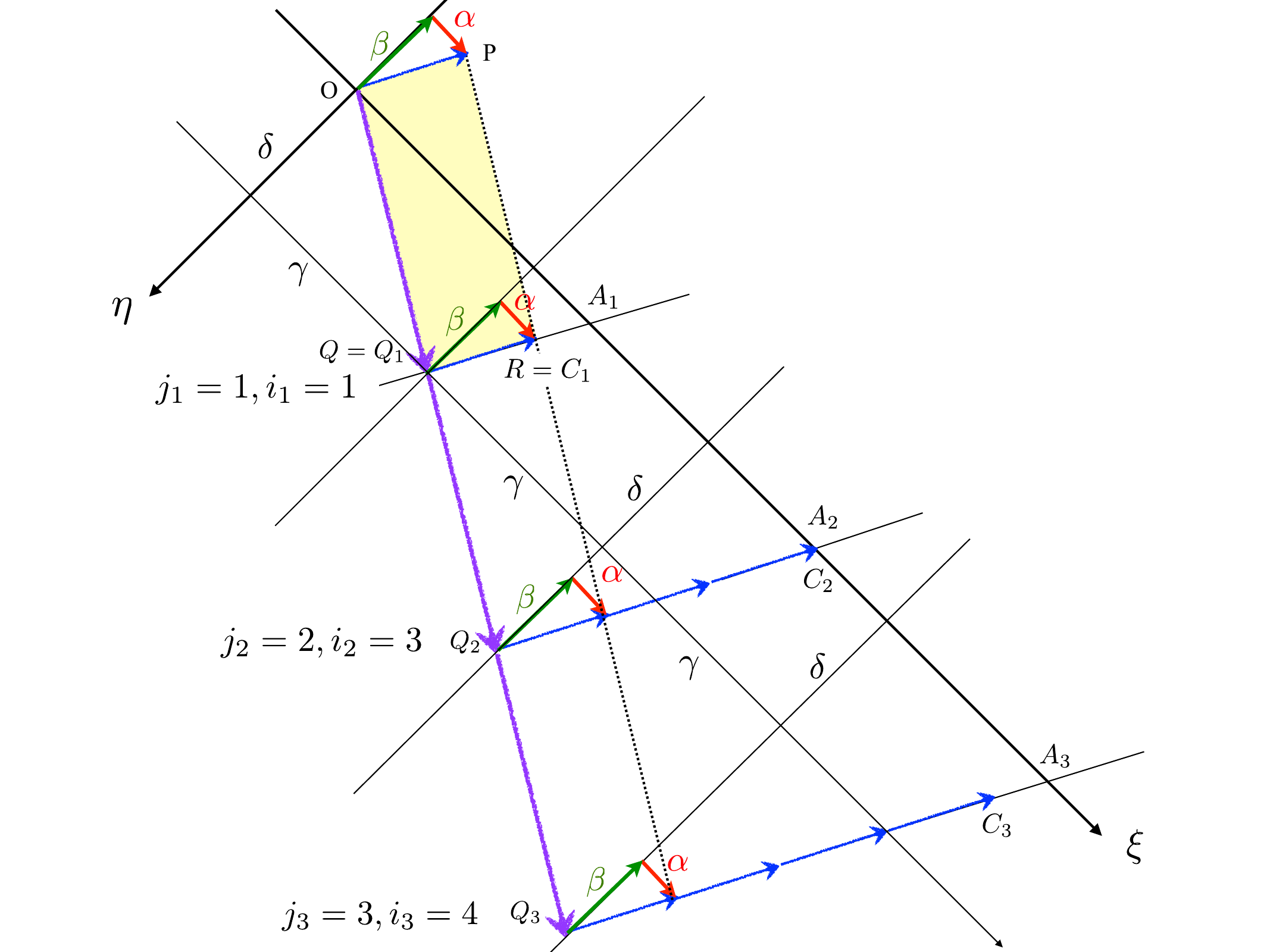}
		\caption{Referenced in the proof of Theorem \ref{minimal cycles}.}
		\label{cyc-theorem}
	\end{center}
\end{figure}
More precisely, 
for any $\xi \geq 0, \eta \geq 0$ we have to prove 
$d(O,\begin{pmatrix} \xi \\ \eta \end{pmatrix}) \geq d(O,\begin{pmatrix} \xi \\ \eta \end{pmatrix} + \begin{pmatrix} \alpha \\ -\beta \end{pmatrix} ) $ under the condition $\eta - \beta \geq 0$. 
This follows from
$\alpha \leq \beta$ by $0 \geq \alpha - \beta 
\; \Rightarrow \; \xi + \eta \geq \xi +\alpha + \eta- \beta 
\; \Rightarrow \; |\xi + \eta | \geq | \xi +\alpha | + |\eta- \beta |
\; \Rightarrow \; d(O,\begin{pmatrix} \xi \\ \eta \end{pmatrix}) \geq d(O, \begin{pmatrix}\xi + \alpha \\ \eta-\beta \end{pmatrix} )$.
By choosing the maximal value of $i$ under the inequality (\ref{GL1})  we obtain $ i = \lfloor \frac{j \cdot \delta}{\beta}\rfloor$.
Therefore the candidates to compute $cyc$ are the endpoints of the vectors 
\begin{equation}\label{vec-x-j-a}
\vec{x}{_j} = j \cdot \begin{pmatrix} \gamma \\\delta \end{pmatrix} + \lfloor \frac{j \cdot \delta}{\beta}\rfloor \cdot \begin{pmatrix} \alpha \\ - \beta \end{pmatrix}  \;\;\;\;\;\; (j \in \Nat). 
 \end{equation} 
 The distance from the origin $O$ to $\vec{x}{_j} $ is $d(O,\vec{x}{_j}) = $
 $j\cdot (\gamma+ \delta)+\lfloor\frac{j \cdot \delta}{\beta}\rfloor \cdot(\alpha-\beta) $. 
For the alternative case $\alpha > \beta$  we look at the symmetric cycloid $\zyk(\beta,\alpha,\delta,\gamma) $ (by interchanging $\alpha$ and $\beta$, as well as $\gamma$ and $\delta$), which is net isomorphic (Theorem \ref{symmetry} c) and therefore has a minimal cycle of the same length. Equation (\ref{vec-x-j-a})  is replaced by Equation (\ref{vec-x-j-b} ):
\begin{equation}\label{vec-x-j-b}
\vec{x}{_j} = j \cdot \begin{pmatrix} \delta \\\gamma \end{pmatrix} + \lfloor \frac{j \cdot \gamma}{\alpha}\rfloor \cdot \begin{pmatrix} \beta \\ - \alpha \end{pmatrix}  \;\;\;\;\;\; (j \in \Nat). 
 \end{equation} 
and we obtain
 $d(O,\vec{x}{_j})   = j \cdot (\gamma + \delta) + \lfloor\frac{j\cdot\gamma}{\alpha}\rfloor \cdot (\beta -\alpha)$ in  case of $\alpha > \beta$. \\
 
 To derive the bound we start with the observation that the length of a cycle is bounded by the number $A$ of transitions. In the case $\alpha  \leq \beta$ it follows with respect to the minimal value of $j$:

 $ cyc = j\cdot (\gamma+ \delta)+\lfloor\frac{j \cdot \delta}{\beta}\rfloor \cdot(\alpha-\beta)   \leq A$ which transforms to 
 $  j\cdot \delta - \lfloor\frac{j \cdot \delta}{\beta}\rfloor \cdot\beta + j \cdot \gamma +  
 \lfloor\frac{j \cdot \delta}{\beta}\rfloor \cdot \alpha  = 
 (j\cdot \delta)\; mod \;\beta + j \cdot \gamma +  
 \lfloor\frac{j \cdot \delta}{\beta}\rfloor \cdot \alpha  
 \leq A$. Since $(j\cdot \delta)\; mod \;\beta +\lfloor\frac{j \cdot \delta}{\beta}\rfloor \cdot \alpha \ \geq 0$ we obtain $j \cdot \gamma  \leq A$ and $j  \leq \frac{A}{\gamma} $. The result for the case $\alpha   \geq \beta$ is proved in a similar way.

 %
 c) We prove $j=1$ in case b) of the theorem under the the additional condition $\gamma  \geq \delta$.
From $\vec{x}{_j}  \geq (0,0)$ we deduce from Equation (\ref{vec-x-j-a}):
\begin{equation}\label{ineq-vec-x-j}
j \cdot \delta - \lfloor \frac{j \cdot \delta}{\beta}\rfloor \cdot \beta  \geq 0
 \end{equation} 
The endpoints of the vectors $\vec{x}{_j}$ are denoted by $C_j$ in Figure \ref{cyc-theorem}. The path from the origin $O$ to  $\vec{x}{_j}$ has the length
$cyc_j := j \cdot (\gamma+\delta) + \lfloor \frac{j \cdot \delta}{\beta}\rfloor \cdot (\alpha-\beta)$ and we next prove $cyc_j  \geq cyc_1$ which shows that $j=1$ is the optimal solution for $cyc$. This is done by the inequality \\
$cyc_j  - cyc_1 = 
j \cdot (\gamma+\delta) + \lfloor \frac{j \cdot \delta}{\beta}\rfloor \cdot (\alpha-\beta) - (\gamma+\delta+ \lfloor \frac{ \delta}{\beta}\rfloor \cdot (\alpha-\beta)) = $ \\
 $(j-1)\cdot \gamma + (j\cdot\delta-\lfloor \frac{j \cdot \delta}{\beta}\rfloor \cdot \beta) - \delta + (\lfloor \frac{j \cdot \delta}{\beta}\rfloor \cdot \alpha - \lfloor \frac{\delta}{\beta}\rfloor \cdot \alpha) +  \lfloor \frac{\delta}{\beta}\rfloor \cdot \beta  \geq \\
 (j-1)\cdot \gamma  - \delta  \geq \gamma - \delta  \geq 0$. \\
 The second summand in the second line is not negative due to the inequality (\ref{ineq-vec-x-j}). This holds obviously for the fourth summand 
 $\lfloor \frac{j \cdot \delta}{\beta}\rfloor \cdot \alpha - \lfloor \frac{\delta}{\beta}\rfloor   \geq 0$. In the last inequalities $j > 0$ and $\gamma  \geq \delta$ is used.
 It remains to prove that negative values of $i$ are not needed to compute $cyc$ in the cases under review. In the same way as before the inequality
 $d(O,\begin{pmatrix} \xi \\ \eta \end{pmatrix})  \leq d(O, \begin{pmatrix}\xi \\ \eta\end{pmatrix} -\begin{pmatrix}\alpha \\ -\beta \end{pmatrix} )$ is proved. This shows, that adding the vector $-(\alpha,-\beta)$ to $(\gamma,\delta)$ cannot decrease the distance from the origin. The cases for $j > 1$ is are similar.\\
For the alternative case $\alpha > \beta$ and $\delta  \geq \gamma$ we look at the symmetric cycloid $\zyk(\beta,\alpha,\delta,\gamma) $ (by interchanging $\alpha$ and $\beta$, as well as $\gamma$ and $\delta$).\\
d) If $\beta | \delta$ then Equation (\ref{vec-x-j-a}) becomes \\
$\vec{x}{_j} = j \cdot \begin{pmatrix} \gamma \\\delta \end{pmatrix} + \frac{j \cdot \delta}{\beta} \cdot \begin{pmatrix} \alpha \\ - \beta \end{pmatrix} = 
\begin{pmatrix} j\cdot \gamma+j\cdot\frac{\delta}{\beta} \cdot \alpha\\ j\cdot \delta+j\cdot\frac{\delta}{\beta} \cdot (-\beta) \end{pmatrix} = 
j \cdot\begin{pmatrix} \gamma+\frac{\delta}{\beta} \cdot \alpha\\ 0\end{pmatrix} $. \\
 Since all the points for different $j$ are on the $\xi$-axis, for $j=1$ we obtain a minimal value of  $cyc = \gamma+\frac{\delta}{\beta} \cdot \alpha$.\\
e) Again, for the alternative case $\alpha > \beta$ and $\alpha  |  \gamma$ we look at the symmetric cycloid. 
\end{proof}

 To illustrate part c) of Theorem \ref{minimal cycles} consider
  the cycloid  $ \mathcal{C}( 1,1,4,1) $. The points 
  $\vec{x} = j \cdot \begin{pmatrix} \gamma \\\delta \end{pmatrix} + i \cdot \begin{pmatrix} \alpha \\ - \beta \end{pmatrix} 
  = \begin{pmatrix} 4 \\ 1 \end{pmatrix} + i \cdot \begin{pmatrix} 1 \\ - 1 \end{pmatrix} 
  = \begin{pmatrix} u \\ v \end{pmatrix} $ are 
  $\begin{pmatrix} 5 \\ 0 \end{pmatrix} $ and
   $\begin{pmatrix} 4 \\ 1 \end{pmatrix} $ for $i = 1$ and $i = 0$, respectively, both resulting in $cyc=u+v= 5$.
   On the other side for the values $i = -1, -2, -3, -4$ we obtain 
   $\begin{pmatrix} 3 \\ 2 \end{pmatrix} $,
   $\begin{pmatrix} 2\\ 3 \end{pmatrix} $,
   $\begin{pmatrix} 1 \\ 4 \end{pmatrix} $  and
   $\begin{pmatrix} 0 \\ 5 \end{pmatrix} $, respectively, with the same result for $cyc$. Hence $i =  \lfloor \frac{\delta}{\beta}\rfloor = \lfloor \frac{1}{1}\rfloor = 1$
   is suffient. The cases for $j > 1$  are similar.

 The pattern of Figure \ref{cyc-theorem} is derived from  the cycloid  $ \mathcal{C}( 1,2,5,3 ) $. The point $C_3$ is computed by the following formula, as derived in the preceeding proof: $\vec{x}{_3} = 3 \cdot \begin{pmatrix} \gamma \\\delta \end{pmatrix} + \lfloor \frac{3 \cdot \delta}{\beta}\rfloor \cdot \begin{pmatrix} \alpha \\ - \beta \end{pmatrix}  =
  3 \cdot \begin{pmatrix} 5 \\ 3 \end{pmatrix} + \lfloor \frac{3 \cdot 3}{2}\rfloor \cdot \begin{pmatrix} 1 \\ - 2 \end{pmatrix}  = \begin{pmatrix} 19 \\ 1  \end{pmatrix} $, leading to 
  $cyc_3 = 20$. The calues for $C_2$ and $C_1 = R$ are $\begin{pmatrix} 13 \\ 0  \end{pmatrix} $ and $\begin{pmatrix} 6 \\ 1  \end{pmatrix} $, respectively.
  
  For the cycloid  $ \mathcal{C}( 2,8,1,4) $ we obtain $cyc = 5$ by Theorem \ref{minimal cycles} c).\\ However, using 
  $  \mathbf{A} \cdot  \begin{pmatrix} i \\ j  \end{pmatrix} = 
  \begin{pmatrix} 2 & 1 \\ -8 & 4 \end{pmatrix}\cdot  \begin{pmatrix} 1 \\ 2 \end{pmatrix} =
  \begin{pmatrix} 4 \\ 0  \end{pmatrix} $ we obtain $cyc = 4$, giving a counter-example to part c) of the theorem.


\section{Conclusion}\label{sec-conclusion}
Using \emph{Cycloid Algebra} a new proof for some important net isomorphisms of cycloids and the problem of equivalence is derived. By the same method also a new proof for the minimal length of a cycloid cycle is obtained, which extends the formula from \cite{Valk-2019}. This approach makes proofs simpler, as otherwise more complicated and combinatorial methods were used.

\bibliographystyle{splncs03}
\bibliography{citations-rv}

\begin{thebibliography}{1}
\providecommand{\url}[1]{\texttt{#1}}
\providecommand{\urlprefix}{URL }

\bibitem{fenske-da}
Fenske, U.: {P}etris {Z}ykloide und {{\"U}}berlegungen zur {V}erallgemeinerung.
  {D}iploma {T}hesis (2008)

\bibitem{Kummer-Stehr-1997}
Kummer, O., Stehr, M.O.: Petri's {A}xioms of {C}oncurrency - a {S}election of
  {R}ecent {R}esults. In: Application and {T}heory of {P}etri {N}ets 1997.
  Lecture Notes in Computer Science, vol. 1248, pp. 195 -- 214.
  Springer-Verlag, Berlin (1997)

\bibitem{Petri-NTS}
Petri, C.A.: {Nets, Time and Space}. Theoretical Computer Science  (153),
  3--48 (1996)

\bibitem{SR:SemNN:87}
Smith, E., Reisig, W.: The semantics of a net is a net -- an exercise in
  general net theory. In: Voss, K., Genrich, J., Rozenberg, G. (eds.)
  Concurrency and Nets. pp. 461--479. Springer-Verlag, Berlin (1987)

\bibitem{Valk-2019}
Valk, R.: Formal {P}roperties of {P}etri's {C}ycloid {S}ystems. Fundamenta
  Informaticae  169,  85--121 (2019)

\bibitem{Valk-2020}
Valk, R.: Circular {T}raffic {Q}ueues and {P}etri's {C}ycloids. In: Application
  and {T}heory of {P}etri {N}ets and {C}oncurrency. Lecture Notes in Computer
  Science, vol. 12152, pp. 176 -- 195. Springer-Verlag, Berlin (2020)

\end{thebibliography}
\end{document}